\documentclass{sig-alternate}
\usepackage[T1]{fontenc}
\usepackage[latin9]{inputenc}
\usepackage{color}
\usepackage{float}
\usepackage{amsmath}
\usepackage{amssymb}
\usepackage{graphicx}

\makeatletter

\floatstyle{ruled}
\newfloat{algorithm}{tbp}{loa}
\providecommand{\algorithmname}{Algorithm}
\floatname{algorithm}{\protect\algorithmname}

\newtheorem{thm}{\protect\theoremname}
\newtheorem{defn}[thm]{\protect\definitionname}
\newtheorem{lem}[thm]{\protect\lemmaname}
\newtheorem{prop}[thm]{\protect\propositionname}

\makeatother

\providecommand{\definitionname}{Definition}
\providecommand{\lemmaname}{Lemma}
\providecommand{\propositionname}{Proposition}
\providecommand{\theoremname}{Theorem}

\begin{document}

\title{Efficient Approximation Algorithms for Computing \emph{k} Disjoint Restricted
Shortest Paths}

\numberofauthors{4} 
%
\author{
\alignauthor
Longkun Guo \\
       \affaddr{College of Mathematics and Computer Science}\\
       \affaddr{Fuzhou University}\\
       \affaddr{China}\\
       \email{lkguo@fzu.edu.cn}
\alignauthor
Kewen Liao \\
       \affaddr{School of Computer Science}\\
       \affaddr{University of Adelaide}\\
       \affaddr{Australia}\\
       \email{kevin.liao@gmail.com}
\alignauthor Hong Shen \\
       \affaddr{School of Computer Science}\\
       \affaddr{ University of Adelaide}\\
       \affaddr{Australia}\\
       \email{hong.shen@cs.adelaide.edu.au}
\and  
\alignauthor Peng Li\\
       \affaddr{Dept of Computer Science and Engineering}\\
       \affaddr{Washington University in St. Louis}\\
       \affaddr{U.S.A.}\\
       \email{pengli@wustl.edu}
}
\maketitle
\begin{abstract}
Network applications, such as multimedia streaming and video conferencing,
impose growing requirements over Quality of Service (QoS), including
bandwidth, delay, jitter, etc. Meanwhile, networks are expected to
be load-balanced, energy-efficient, and resilient to some degree of
failures. It is observed that the above requirements could be better
met with multiple disjoint QoS paths than a single one. Let $G=(V,\, E)$
be a digraph with nonnegative integral cost and delay on every edge,
$s,\, t\in V$ be two specified vertices, and $D\in\mathbb{Z}_{0}^{+}$
be a delay bound (or some other constraint), the \emph{$k$ Disjoint
Restricted Shortest Path} ($k$\emph{RSP})\emph{ Problem} is computing
$k$ disjoint paths between $s$ and $t$ with total cost minimized
and total delay bounded by $D$. Few efficient algorithms have been
developed because of the hardness of the problem.

In this paper, we propose efficient algorithms with provable performance
guarantees for the $k$RSP problem. We first present a pseudo-polynomial-time
approximation algorithm with a bifactor approximation ratio of $(1,\,2)$,
then improve the algorithm to polynomial time with a bifactor ratio
of $(1+\epsilon,\,2+\epsilon)$ for any fixed $\epsilon>0$, which
is better than the current best approximation ratio $(O(1+\gamma),\, O(1+\frac{1}{\gamma})\})$
for any fixed $\gamma>0$ \cite{orda2004efficient}. To the best of
our knowledge, this is the first constant-factor algorithm that almost
strictly obeys the constraint for the $k$RSP problem. \end{abstract}

\keywords{
$k$ disjoint restricted shortest path, bifactor approximation algorithm,
auxiliary graph, cycle cancellation.}

\section{Introduction}

\subsection{Background}
The disjoint quality of service (QoS) path problem is a generalization
of the shortest QoS path problem and has broad applications in networking,
data transmission, etc. In data networks, many applications, such
as video streaming, video conferencing, and on-demand video delivery,
have several QoS requirements, which require the routing between source
and destination nodes to simultaneously satisfy several QoS constraints,
such as bandwidth and delay. In those applications, a single link
might not provide adequate bandwidth, and multiple disjoint QoS paths
are often necessary. Given cost and delay as QoS constraints, the
$k$\emph{ disjoint QoS path problem} can be defined as follows.
\begin{defn}
(\emph{The} \emph{$k$ disjoint QoS path problem}) Given a digraph
$G=(V,\, E)$, a pair of distinct vertices $s,\, t\in V$, a cost
function $c:\, E\rightarrow\mathbb{Z}_{0}^{+}$, a delay function
$d:\, E\rightarrow\mathbb{Z}_{0}^{+}$, and a given delay bound $D\in\mathbb{Z}_{0}^{+}$,
the\emph{ $k$ disjoint QoS path problem} is to compute $k$ disjoint
$st$-paths $P_{1},\dots,P_{k}$, i.e., $E(P_{i})\cap E(P_{j})=\emptyset$
for every $i\neq j\in\{1,\,\dots,\, k\}$, such that $d(P_{i})\leq D$
for each $i=1,\dots,k$, and the total cost of the $k$ disjoint paths
is minimized.
\end{defn}
This problem is \emph{NP}-hard even when all edges of $G$ has a cost
of zero \textcolor{black}{\cite{li1989cft}}. The hardness result
indicates that it is impossible to develop exact or  polynomial-time approximation
algorithms that strictly obey the delay constraint for the\emph{ }$k$
disjoint QoS paths problem unless \emph{P=NP}. An alternative method
is to compute $k$ disjoint paths with total cost minimized
and a total delay bounded by $D$ (equal to $kD$ in Definition 1),
and then route the packages via the $k$ paths according to their
urgency priority, i.e., routing urgent packages via paths of low delay
whilst deferrable ones via paths of high delay. Then, \emph{the }$k$
\emph{disjoint Restricted Shortest Path} ($k$RSP) problem arises
as in the following:
\begin{defn}
(\emph{The} \emph{$k$ disjoint Restricted Shortest Path problem},
$k$RSP) Given a digraph $G=(V,\, E)$, a pair of distinct vertices
$s,\, t\in V$, a cost function $c:\, E\rightarrow\mathbb{Z}_{0}^{+}$,
a delay function $d:\, E\rightarrow\mathbb{Z}_{0}^{+}$, and a delay
bound $D\in\mathbb{Z}_{0}^{+}$, the\emph{ $k$} (edge) disjoint Restricted
Shortest Paths (\emph{k}RSP)\emph{ }problem is to calculate $k$ disjoint
$st$-paths $P_{1},\,\dots,\, P_{k}$, such that $E(P_{i})\cap E(P_{j})=\emptyset$
 for any $i\neq j\in\{1,\,\dots,\, k\}$, $\underset{i=1,\dots,k}{\sum}d(P_{i})\leq D$,
and the cost of the $k$ paths is minimized.
\end{defn}
Previous works on the $k$RSP problem are mainly bifactor approximation
algorithms. An algorithm ${\cal A}$ is a bifactor $\left(\alpha,\,\beta\right)$-approximation
algorithm for the $k$RSP problem if and only if for every instance
of $k$RSP\emph{, }${\cal A}$ runs in polynomial time and outputs
$k$ disjoint $st$-paths, and the total delay and the total cost
of the computed $k$ disjoint paths are bounded by $\alpha D$ and
$\beta C_{OPT}$, respectively, where $C_{OPT}$ is the cost of an optimal solution to $k$RSP, $\alpha$ and $\beta$ are positive
constants. Note that a single factor $\beta$-approximation is identical
to a bifactor $\left(1,\,\beta\right)$-approximation, and we will
use them interchangeably in the rest of the paper.

The $k$RSP problem is a fundamental problem in multipath routing that has a variety of benefits including fault tolerance,
increased bandwidth, load balance and etc. Although multipath routing is not widely
deployed in practice because of the difficulty of developing efficient
algorithms, research interest focus on multipath routing is growing. The reason is  the development of software defined networking
(SDN),  which seems a wonderful place to deploy multipath routing. Controllers
therein have global information of the network and stronger computational
ability, which make it possible to deploy some complicated routing
algorithm into the network. Our results show that $k$RSP admits efficient
approximation algorithm theoretically. This may help to boost the
development of multipath routing in SDN.

$k$RSP has theoretical values beyond multipath routing. It is an interesting
bicriteria optimization problem, because the shortest path problem
is a fundamental problem in the area of combinatorial optimization. In addition, our method on $k$RSP
might be applied to other related problems.

\subsection{Related Work}

In general,\textcolor{black}{{} $k$RSP is a budgeted optimization problem:
given the maximum delay constraint, find $k$ disjoint paths of minimum
cost, which has been well-investigated recently. Efficient algorithms
have been developed for budgeted matching and budgeted matroid intersection
\cite{berger2011budgeted,chekuri2011multi}.} However, those methods
can not be adopted to solve $k$RSP because $k$RSP cannot be modeled
as matching or \textcolor{black}{matroid intersection. Another budgeted
network design problem, }the\emph{ shallow-light Steiner tree (SLST)
}problem\textcolor{black}{, is }attracting\textcolor{black}{{} lots
of research interest.} \emph{SLST} is to compute a minimum cost tree
spanning a set of given terminals, such that the cost of the computed
tree is minimized and the delay from a specified terminal to every
other terminal is not larger than $D$. This problem can not be approximated
better than factor $(1,\,\gamma\log^{2}n)$ for some fixed $\gamma>0$
unless $NP\subseteq DTIME(n^{\log\log n})$ \cite{khandekar2013some}.
Further, although the inapproximability result doesn't exclude the
existence of polylogarithmic factor approximation algorithms for \emph{SLST},
to the best of our knowledge, no such algorithms within polynomial
time complexity have been developed. The algorithm with the best ratio
is a long standing result by Charikar et al, which is a polylogarithmic
approximation algorithm that runs in quasi-polynomial time, i.e.,
a factor-$O(\log^{2}t)$ approximation algorithm within time complexity
$n^{O(\log t)}$\cite{charikar1998approximation}. Due to the difficulty
in designing single factor approximation algorithms, bifactor approximation
algorithms have been investigated. Hajiaghayi et al presented an $(O(\log^{2}t),\, O(\log^{4}t))$-approximation
algorithm that runs in polynomial time \cite{hajiaghayi2006approximating}.
Besides, Kapoor and Sarwat designed an approximation algorithm with
bifactor $(O(\frac{p\log t}{\log p}),O(\frac{\log t}{\log p}))$,
where $p$ is an input parameter \cite{kapoor2007bounded}. Their
algorithm is an approximation algorithm that improves the cost of
the tree, and is with bifactor $(O(t),\, O(1))$ when $p=t$ \cite{kapoor2007bounded}.
It is even more interesting that,  for the special case $S=V$, Charikar
et al's ratio of $O(\log^{2}n)$ \cite{charikar1998approximation}
(with quasi-polynomial time complexity $n^{O(\log n)}$) is still
the best single factor ratio. In addition, the existing bifactor approximations
for \emph{SLST} are also the known currently best approximations for
the special case.

Special cases of $k$RSP have also been studied. When the delay constraint
is removed, this problem is reduced to the min-sum disjoint path problem
of calculating $k$ disjoint paths with the total cost minimized.
This problem is known polynomially solvable \cite{suurballe1974dpn}.
Moreover, when $k=1$, the problem reduces to the single restricted
shortest path (RSP) problem, which is known as a basic QoS routing
problem \cite{garey1979computers} and admits fully polynomial-time
approximation scheme (FPTAS) \cite{garey1979computers,lorenz2001simple}.
The (single) QoS path problem with multiple constraints is still attracting
considerable research interests. Xue et al. recently proposed a ($1+\epsilon$)-approximation
algorithm \cite{xue2008polynomial}. \textcolor{black}{W}hen\textcolor{black}{{}
the delay constraint is on each single path, and the cost on every
edge is zero, the $k$RSP problem reduces to the length-bounded disjoint
path problem of finding two disjoint paths with the length of each
path constrained by a given bound. This problem is a variant of the
Min-Max problem, which is to find two disjoint paths with the length
of the longer one minimized. Both problems are known to be }\textcolor{black}{\emph{NP-}}\textcolor{black}{complete
\cite{li1989cft} with the best possible approximation ratio of 2
in digraphs \cite{li1989cft}, which can be achieved by applying the
algorithm for the min-sum problem in \cite{suurballe1974dpn,suurballe1984qmf}.
In contrast, the min-min problem of finding two disjoint paths with
the length of the shorter path minimized is also NP-complete but doesn't
admit $k$-approximation for any $k\geq1$ \cite{bhatia2006finding,algorithmicaGuoS13,xu2006caa}.
The problem remains NP-complete and admits no polynomial-time approximation
scheme in planar digraphs \cite{guo2012tcs}. }

A closely related problem, the $k$ disjoint bi-constrained path problem
($k$BCP), targets $k$ disjoint $st$-paths that satisfy the given
cost constraint $\underset{i=1,\dots,k}{\sum}c(P_{i})\leq C$ and
the given delay constraint $\underset{i=1,\dots,k}{\sum}d(P_{i})\leq D$.
For the $k$BCP problem, approximation algorithms with a bifactor
approximation ratio of $(1+\beta,\,\max\{2,\,1+\ln\frac{1}{\beta})$
or a single factor ratio of $O(\ln n)$ (i.e. bifactor approximation
ratio $(1,\, O(\ln n))$) have been developed for general $k$ in
\cite{GuoSL13}, where $\beta>0$ is any fixed positive real number.
Apparently, $k$BCP is a weaker version of $k$RSP, and hence all
approximations of $k$RSP can be adopted to solve $k$BCP, but not
the other way around.

The $k$RSP problem itself has attracted considerable research interest of
a number of computer scientists. \cite{chao2007new} and \cite{orda2004efficient}
have achieved bifactor ratios of $(1+\frac{1}{r},\, r(1+\frac{2(\log r+1)}{r})(1+\epsilon))$
and $(1+\frac{1}{r},\,1+r)$ for $k=2$, respectively. Based on LP-rounding
technology, an approximation with bifactor ratio $(2,\,2)$ has been
developed in \cite{DBLP:conf/faw/Guo14}. To the best of our knowledge,
however, no algorithm has achieved a constant single factor approximation
ratio for decades. Our algorithm is the first one showing that $k$RSP
admits a almost constant single factor approximation ratio theoretically.

\subsection{Our Results}

The main results of this paper are summarized in this section with
proofs detailed in latter sections.
\begin{lem}
\label{lem:mainre}The $k$RSP problem admits an approximation algorithm
with an approximation ratio of $(1,\,2)$ and runtime $O((\sum c(e))^{2}\sum d(e)(n^{4.5}C_{OPT}^{4.5}DL))$,
where $C$ is the cost of an optimal solution and $L$ is the maximum
length of the input.
\end{lem}
We note that the algorithm is with pseudo-polynomial time complexity,
because the above formula of the time complexity involves $C_{OPT}$, $D$,
and edge costs and edge delays of the graph. However, by applying
the traditional technique for polynomial time approximation scheme
design as in \cite{garey1979computers}, we can immediately obtain
a polynomial time algorithm with bifactor approximation ratio ($1+\epsilon_{1},\,2+\epsilon_{2}$).
Below is the main idea of the technique: For any constant $\epsilon_{1},\,\epsilon_{2}>0$,
set the delay and cost of every edge $e$ to $\left\lfloor \frac{d(e)}{\frac{\epsilon_{1}D}{n}}\right\rfloor $
and $\left\lfloor \frac{c(e)}{\frac{\epsilon_{2}C_{OPT}}{n}}\right\rfloor $,
respectively. Then the time complexity of the algorithm becomes polynomial
time, since

\[O((\sum c(e))^{2}\sum d(e)(n^{4.5}C_{OPT}^{4.5}DL))\]
\[\leq O(\sum^{2}\left\lfloor \frac{c(e)}{\frac{\epsilon_{2}C_{OPT}}{n}}\right\rfloor \sum\left\lfloor \frac{d(e)}{\frac{\epsilon_{1}D}{n}}\right\rfloor (n^{4.5}(\frac{n}{\epsilon_{2}})^{4.5}\frac{n}{\epsilon_{1}}L))\]
\[\leq O(m^{3}n^{12}\frac{1}{\epsilon_{1}^{2}\epsilon_{2}^{6.5}})\].
\begin{thm}
For any constants $\epsilon_{1},\,\epsilon_{2}>0$, the $k$RSP problem
admits a polynomial time algorithm with an approximation ratio $(1+\epsilon_{1},\,2+\epsilon_{2})$.
\end{thm}
The time complexity of the algorithm can be analyzed better. Anyhow,
to the best of our knowledge, this is the first constant factor approximation
algorithm for the $k$RSP problem that runs in polynomial time, and
almost strictly obeys the delay constraint. Note that we can set $\epsilon_{1}=\epsilon_{2}=\epsilon$,
and the approximation ratio will be $(1+\epsilon,\,2+\epsilon)$,
as we claimed in the abstract.

Our pseudo polynomial time approximation algorithm of Lemma \ref{lem:mainre}
mainly consists of two phases. The first phase is to compute a not-too-bad
solution for the $k$RSP problem by a simple LP-rounding algorithm
as in \cite{DBLP:conf/faw/Guo14}, whose performance guarantee has
been shown as in the following:
\begin{lem}
\label{lem:guofaw} The $k$RSP problem admits an algorithm, such
that for any of its output solutions, there exists a real number $0\leq\alpha\leq2$
such that the delay-sum and the cost-sum of the solution are bounded
by $\alpha D$ and $(2-\alpha)*C_{OPT}$, respectively.
\end{lem}
Note that $\alpha$ differs for different instances, i.e., the algorithm
might return a solution with cost $2C_{OPT}$ and delay $0$ for some instances,
whilst a solution with cost 0 and delay $2D$ for other instances.
So the bifactor approximation ratio for the algorithm is actually
$(2,2)$. The second phase of our algorithm, as the main task of this
paper, is to improve the approximation ratio $(2,\,2)$ to $(1+\epsilon,\,2+\epsilon)$
in pseudo polynomial time.

The main techniques involved in the second phase include the cycle
cancellation method, LP-rounding, and some graph transformations techniques.
We would like to remark here that the cycle cancellation method is
a classic framework that has been used to solve numerous problems
related to shortest path, minimum cost flow and etc \cite{ahuja1993network,GuoSL13,orda2004efficient}.
However, to the best of our knowledge, in previous works the cycles
used are computed in a graph which allows either negative-cost edges
or negative-delay edges, but not both. Because no polynomial algorithm
is known for computing a \emph{best} cycle for cycle cancellation
in a graph allowing both negative cost and delay. This paper not only
gives an approximation algorithm for the $k$RSP problem, but also
enhances the cycle cancellation method by giving a novel algorithm
of computing bicameral cycle, which is a \emph{good-enough} cycle
for cycle cancellation in a graph where both negative-cost edges and
negative-delay edges are allowed (The formal definition of bicameral
cycles will be given later). That is, our algorithm can hopefully
improve the approximation ratios of other bicriteria optimization
problems (or namely budgeted optimization problems) in which the cycle
cancellation method can be suitably applied.

The remainder of the paper is organized as below: Section 2 introduces
cycle cancellation, Section 3 gives the definition of bicameral cycles,
and then shows that using bicameral cycles in cycle cancellation can
result in a desired approximation ratio $(1+\epsilon_{1},\,2+\epsilon_{2})$.
This is the first tricky task of this paper. Section 4 gives the algorithm
that actually computes a bicameral cycle in polynomial time, which
is the most tricky task of this paper. Section 5 concludes this paper.

\section{The cycle cancellation method for $k$RSP}

This section will give the key idea of an improved approximation for
$k$RSP based on the cycle cancellation method. We shall first state
our version of the cycle cancellation method (and point out the difference
comparing to previous versions), then give the improved algorithm
that is an enhancement to the cycle cancellation technique.

\subsection{The Cycle Cancellation Method}

We would like to start from some definitions and notations. Let $E_{1}$
and $E_{2}$ be two set of edges. Then $E_{1}\oplus E_{2}$ denotes
the edge set $E_{1}\cup E_{2}\setminus\{e(u,\, v)\vert\{e(u,\, v),\, e(v,\, u)\}\subseteq E_{1}\cup E_{2}\}$,
i.e., $E_{1}\cup E_{2}$ except pairs of parallel edges in opposite
direction therein. Let $P$ be a $st$-path and $\overline{P}=\{e'(v,u)\vert e(u,v)\in P\}$,
i.e., $\overline{P}$ is $P$ but with
the direction of each edge reversed. Moreover, the cost and delay of the edges of $\overline{P}$
are negatived, i.e., $c(e'(v,\, u))=-c(e(u,\, v))$ and $d(e'(v,\, u))=-d(e(u,\, v))$
for every edge $e'(v,\, u)\in\overline{P}$.
\begin{defn}
\label{def:Residual-graph}(Residual graph) A residual graph $\widetilde{G}=G_{res}(P_{1},\,\dots,\, P_{k})$,
with respect to $G$ and $P_{1},\,\dots,\, P_{k}\subset G$, is graph
$G\cup(\underset{i=1,\,\dots,\, k}{\cup}E(\overline{P_{i}}))\setminus\underset{i=1,\,\dots,\, k}{\cup}E(P_{i})$,
i.e., graph $G$ with the direction of edges of $P_{1},\,\dots,\, P_{k}$
reversed, and their cost and delay negatived. %
\footnote{Note that $\widetilde{G}$ can contain pairs of parallel edges in
the same direction with different costs and delays. Thus, $\widetilde{G}$
might be a multigraph. %
}
\end{defn}
For notation briefness, we use $\widetilde{G}$ instead of $G_{res}(P_{1},\,\dots,\, P_{k})$
while no confusion arises. In literature \cite{GuoSL13}
and \cite{orda2004efficient}, the authors also employed the cycle cancellation
method. However, in their constructed residual graph, the cost of
reversed edge $e'(v,u)$ is sat  $c(e'(v,u))=0$ instead of $c(e'(v,u))=-c(e(u,v))$,
such that the edges in their residual graphs are with nonnegative
cost. Then, the minimum-mean-cycle algorithm can be applied therein,
and hence a \emph{best} cycle for cycle cancellation, i.e., $O$ with
$\frac{d(O)}{c(O)}$ minimized, can be computed in polynomial time
\cite{korte2012combinatorial}. Note that for a more complicated residual
graph as in Definition \ref{def:Residual-graph}, the method above
can not be applied to compute a best cycle any longer.

The cycle cancellation method is based on the following proposition
which can be derived from the flow theory:
\begin{prop}
\label{prop:cyclever}Let $P_{1},\,\dots,\, P_{k}$ be $k$ disjoint
$st$-paths in $G$ and $O_{1},\,\dots,\, O_{h}$ be a set of edge-disjoint
cycles in residual graph $\widetilde{G}$. Then $\{P_{1},\,\dots,\, P_{k}\}\oplus\{O_{1},\,\dots,\, O_{h}\}$
contains $k$ disjoint $st$-paths in $G$. \end{prop}
\begin{proof}
Following the flow theory, we need only to show that $\{P_{1},\,\dots,\, P_{k}\}\oplus\{O_{1},\,\dots,\, O_{h}\}$contains
an integral $st$-flow of value $k$ and the flow is a valid flow
in $G$.

Firstly, $\{P_{1},\,\dots,\, P_{k}\}\oplus\{O_{1},\,\dots,\, O_{h}\}$
contains an integral flow of value $k$ that goes from $s$ to $t$.
That is because $s$ and $t$ are respectively with degree $k$ and
$-k$ in $P_{1}\cup\,\dots\,\cup P_{k}$ whilst every vertex are with
degree 0 in $O_{1},\,\dots,\, O_{h}$. So $s$ and $t$ are respectively
with degree $k$ and $-k$ while every other vertex is with degree
0 in $\{P_{1},\,\dots,\, P_{k}\}\oplus\{O_{1},\,\dots,\, O_{h}\}$.

Secondly, $\{P_{1},\,\dots,\, P_{k}\}\oplus\{O_{1},\,\dots,\, O_{h}\}$
is a valid flow in $G$. On one hand, every edge appears at most once
in the flow of$\{P_{1},\,\dots,\, P_{k}\}\oplus\{O_{1},\,\dots,\, O_{h}\}$
(i.e. satisfies the capacity constraint). That is because every edge
appears at most once in $\{P_{1},\,\dots,\, P_{k}\}\oplus\{O_{1},\,\dots,\, O_{h}\}$,
since$\{P_{1},\,\dots,\, P_{k}\}$ and $\{O_{1},\,\dots,\, O_{h}\}$
shares no common edges, and the paths of $\{P_{1},\,\dots,\, P_{k}\}$
are mutually edge-disjoint and so are the cycles of $\{O_{1},\,\dots,\, O_{h}\}$.
On the other hand, $\{P_{1},\,\dots,\, P_{k}\}\oplus\{O_{1},\,\dots,\, O_{h}\}$
contains only edges of $G$, since every edge of $\{O_{1},\,\dots,\, O_{h}\}\setminus G$
has a parallel opposite counterpart in $\{P_{1},\,\dots,\, P_{k}\}$.
\end{proof}
Intuitively, $\{P_{1},\,\dots,\, P_{k}\}\oplus\{O_{1},\,\dots,\, O_{h}\}$
is to replace the edges of $P_{1},\,\dots,\, P_{k}$, which have parallel
opposite counterpart in $\{O_{1},\,\dots,\, O_{h}\}$, by the edges
of $\{O_{1},\,\dots,\, O_{h}\}\setminus\{\overline{P}_{1},\,$ $\dots,\,\overline{P_{k}}\}$.
Such an edge replacement with cycles is called cycle cancellation.

\subsection{Cycle Cancellation with Bicameral Cycles for $k$RSP}

This subsection shall apply the cycle cancellation method to improve
a solution of the $k$RSP problem. Let $P_{1},\,\dots,\, P_{k}$ be
a solution resulting from the first phase with large delay, say $\sum_{i=1}^{k}d(P_{i})>D$.
According to Proposition \ref{prop:cyclever}, if we could find a
set of edge-disjoint negative delay cycles $O_{1},\,\dots,\, O_{h}$,
then we could decrease the delay of the solution to $k$RSP by using
$\{P_{1},\,\dots,\, P_{k}\}\oplus\{O_{1},\,\dots,\, O_{h}\}$ as the
$k$ disjoint paths. Naturally, a question arises:

\emph{Does there always exist cycles with negative delay in $\widetilde{G}$
when $\sum_{i=1}^{k}d(P_{i})>D$ holds?}

To answer  this question, the following proposition is necessary:
\begin{prop}
\label{prop:cycle}Let $P_{1}^{*},\,\dots,\, P_{k}^{*}$ be a minimum
cost solution to the $k$RSP problem that satisfies the delay constraint,
and $P_{1},\,\dots,\, P_{k}$ be $k$ disjoint $st$-paths. Then $\{P_{1}^{*},\,\dots,\, P_{k}^{*}\}\oplus\{\overline{P_{1}},\,\dots,\,\overline{P_{k}}\}$
is exactly a set of edge-disjoint cycles.
\end{prop}
The key idea to prove the proposition is that: Every vertex in graph
$\{P_{1}^{*},\,\dots,\, P_{k}^{*}\}\oplus\{\overline{P_{1}},\,\dots,\,\overline{P_{k}}\}$
is with degree 0, so $\{P_{1}^{*},\,\dots,\, P_{k}^{*}\}\oplus\{\overline{P_{1}},\,\dots,\,\overline{P_{k}}\}$
is composed by exactly a set of cycles. We note that a generalized
version (for network flow) of Proposition \ref{prop:cycle} has appear
in \cite{ahuja1993network}, so we omit the detailed proof here. Let
$\{O_{1},\dots,O_{h}\}$ be the set of cycles of $\{P_{1}^{*},\,\dots,\, P_{k}^{*}\}\oplus\{\overline{P_{1}},\,\dots,\,\overline{P_{k}}\}$.
Note that $\sum_{i=1}^{k}d(P_{i})>\sum_{i=1}^{k}d(P_{i}^{*})$, so
$\sum_{i=1}^{h}d(O_{i})=\sum_{i=1}^{k}d(P_{i}^{*})-\sum_{i=1}^{k}d(P_{i})<0$
holds, and hence at least one cycle of $\{O_{1},\dots,O_{h}\}$ is
with negative delay. Hence, the previous question has a positive answer,
formally as below:
\begin{lem}
\label{lem:existenceofbc}If \emph{$\sum_{i=1}^{k}d(P_{i})>D$,} then
there exists at least one cycle with negative delay in $\widetilde{G}$\emph{.}
\end{lem}
From Proposition \ref{prop:cyclever} and Proposition \ref{prop:cycle},
an idea for the improved algorithm is to find the set of cycles $\{O_{1},\dots,O_{h}\}$
and use them to improve the $k$ disjoint paths $P_{1},\,\dots,\, P_{k}$
to an optimal solution $P_{1}^{*},\,\dots,\, P_{k}^{*}$. However,
it is hard to identify all the cycles $\{O_{1},\,\dots,\, O_{h}\}$
exactly. So the key idea of the improved algorithm is generally a
greedy approach, which repeats computing a ``best'' cycle (i.e.
cycle $O$ with optimum $\frac{d(O)}{c(O)}$), and using it to improve
the current $k$ disjoint paths towards an optimal solution, until
a desired solution is obtained. However, a ``best'' cycle is probably
\emph{NP}-hard to compute. So our algorithm use a \emph{bicameral}
cycle instead. Let $P_{1},\,\dots,\, P_{k}$ be the current solution
to the $k$RSP problem. Then the main steps of our algorithm roughly
proceeds as below:

\quad{}\emph{ }\textbf{\emph{While}}\emph{ $\sum_{i=1}^{k}d(P_{i})>D$
}\textbf{\emph{do}}\emph{ }

\quad{}\emph{ }\quad{}\emph{ Compute a }bicameral\emph{ cycle $O$; }

\quad{}\emph{ }\quad{}\emph{ Set $\{P_{1},\,\dots,\, P_{k}\}\leftarrow\{P_{1},\,\dots,\, P_{k}\}\oplus O$. }

\section{Proof of Lemma 3}

This section will first give the formal definition of bicameral cycles
as well as the formal layout of our algorithm, then show that, by
using the defined bicameral cycles, our algorithm can output a solution
with a bifactor ratio $(1,\,2)$, and terminate in pseudo polynomial
time.

\subsection{The Improved Algorithm}

Let $P_{1},\,\dots,\, P_{k}$ be a feasible solution to  $k$RSP. W.l.o.g., assume $\sum_{i=1}^{k}d(P_{i})>D$. Let $\Delta D=D-\sum_{i=1}^{k}d(P_{i})$
and $\Delta C=C_{OPT}-\sum_{i=1}^{k}c(P_{i})$, where $C_{OPT}$ is the cost of
an optimal solution to $k$RSP. Let $O$ be a cycle in
the residual graph wrt $P_{1},\,\dots,\, P_{k}$. Then because the
meaning of $\frac{d(O)}{c(O)}$ depends on positive or negative $d(O)$
(or $c(O)$). So we have three types of bicameral cycles accordingly,
as in the following:
\begin{defn}
\label{bicy}(\emph{bicameral} \emph{cycles})
\begin{enumerate}
\item If $d(O)<0$ and $c(O)\leq0$ or $d(O)\leq0$ and $c(O)<0$, then
$O$ is a (type-0) \emph{bicameral} \emph{cycle};
\item Otherwise,

\begin{enumerate}
\item If $d(O)<0$ and $0<c(O)\leq C_{OPT}$, then $O$ is a (type-1) bicameral
cycle iff $\frac{d(O)}{c(O)}\leq\frac{\Delta D}{\Delta C}$, where
$\Delta D=D-\sum_{i=1}^{k}d(P_{i})$ and $\Delta C=\sum_{i=1}^{k}c(P_{i}^{*})-\sum_{i=1}^{k}c(P_{i})$;
\item If $d(O)\geq0$ and $-C_{OPT}\leq c(O)<0$, then $O$ is a (type-2) bicameral
cycle iff $\frac{d(O)}{c(O)}\geq\frac{\Delta D}{\Delta C}$.
\end{enumerate}
\end{enumerate}
\end{defn}
Clearly, a type-0 bicameral cycle is the first choice for our algorithm,
because it can be used to decrease the delay of the $k$ disjoint
paths without any cost increment according to our algorithm. Differently,
using a type-1 (2) cycle $O$ for cycle cancellation will decrease
delay (cost), but increase the cost (delay) of the current solution.
Since computing a bicameral cycle is one of two tricky tasks
of this paper, the algorithm will be deferred to Section 4. The formal
layout of our algorithm is as in Algorithm \ref{alg:An-improved-algorithm}.

\begin{algorithm}
\textbf{Input:} Graph $G$, specified vertices $s$ and $t$, a cost
function $c(e)$ and a delay function $d(e)$ on edge $e$, and a
delay bound $D\in\mathbb{Z}_{0}^{+}$;

\textbf{Output}: $\{P_{1},\dots,P_{k}\}$, the set of edges that compose
$k$ disjoint $st$-paths.
\begin{enumerate}
\item Compute $k$ disjoint $st$-paths $P_{1},\dots,P_{k}$ by the LP-rounding
algorithm as in \cite{DBLP:conf/faw/Guo14};
\item \textbf{While} $\sum_{i=1}^{k}d(P_{i})>D$ \textbf{do}

\begin{enumerate}
\item \textbf{If} there exists no bicameral cycle in $\widetilde{G}$, \textbf{then}
return ``infeasible''; /{*}The instance is infeasible.{*}/
\item Compute a bicameral cycle $O$;
\item $\{P_{1},\dots,P_{k}\}\leftarrow\{P_{1},\dots,P_{k}\}\oplus O$;
\end{enumerate}
\item Return $\{P_{1},\dots,P_{k}\}$.
\end{enumerate}
\protect\caption{\label{alg:An-improved-algorithm}An improved algorithm for $k$RSP.}
\end{algorithm}

We would like to give some remarks on the definition of bicameral
cycles before presenting the ratio proof. Firstly, a bicameral cycle $O$
is not a ``best'' cycle, it suffices only $\frac{d(O)}{c(O)}\leq\frac{\Delta D}{\Delta C}$ (type 1 bicameral cycle for example). There may be other far better cycles for cycle cancellation, but it takes too long to compute. Secondly,
a type-1 bicameral cycle $O_{1}$ must satisfy an additional constraint
$0<c(O_{1})\leq C_{OPT}$, and a type-1 bicameral cycle $O_{2}$ must satisfy
$-C_{OPT}\leq c(O)<0$. This makes the definition of bicameral cycle more
complicated, but it is essential. Since without the constraint, the
ratio of the algorithm will become $(1+\alpha,\,1+\frac{1}{\alpha})$.
In this case, the cost of the solution resulting from the algorithm
could be very large when $\alpha$ is a small number, say $\alpha=\frac{1}{D}$
(see Figure \ref{fig:A-counter-example} for example).

\begin{figure*}
\begin{centering}
\includegraphics{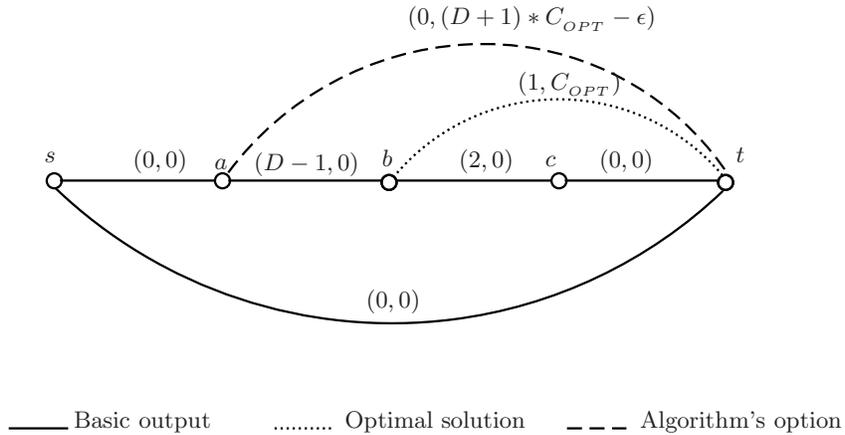}
\par\end{centering}

\protect\caption{\label{fig:A-counter-example}An example for execution of Algorithm
\ref{alg:An-improved-algorithm} without consider the constraint on
the cost: $sabct$ and $st$ is the solution resulted from the simple
approximation algorithm in \cite{lorenz2001simple,DBLP:conf/faw/Guo14};
$sat$ and $st$ would the output of the algorithm, with cost $C_{OPT}*(D+1)-\epsilon$
and delay $0$; while the optimal solution is $sabt$ and $st$, with
cost $C_{OPT}$ and delay $D$.}
\end{figure*}

Because of the complicated constraints on type-1 and type-2 bicameral
cycles, and the fact that the residual graph defined in this paper
allows both negative cost and delay on edges, previous techniques
as in \cite{DBLP:conf/faw/Guo14,orda2004efficient,GuoSL13} are not
suitable for computing a bicameral cycle. Anyhow, this paper figures
out a constructive method for computing bicameral cycles, which will
be shown in Section 4.

\subsection{Proof of Approximation Ratio and Time Complexity }

Assume that the algorithm terminates in $f$ iterations. Without loss
of generality, assume that $f>1$. Let $O_{i}$ be a \emph{bicameral}
cycle in the $i$th iteration, and $D_{i}$ and $C_{i}$ be the delay
and cost of the current solution, respectively. Let $\Delta D_{i}=D-D_{i}$, $\Delta C_{i}=C_{OPT}-C_{i}$,
and $r_{i}=\frac{\Delta D_{i}}{\Delta C_{i}}$. The ratio of Algorithm
\ref{alg:An-improved-algorithm} is as below:
\begin{lem}
\label{lem:ratio}If the $k$RSP problem is feasible, Algorithm \ref{alg:An-improved-algorithm}
outputs a solution with delay strictly bounded by $D$ and cost bounded
by $2C_{OPT}$.\end{lem}
\begin{proof}
Algorithm \ref{alg:An-improved-algorithm} terminates only when the
delay constraint is satisfied. Then from Lemma \ref{lem:existenceofbc},
the lemma is obviously true for the delay constraint part. For the
cost part, according to the definition of bicameral cycles, the cost
augmentation in the last iteration, i.e., the $f$th iteration, is
at most $C_{OPT}$. Then we need only to show that $C_{f-1}\leq C_{OPT}$. Below
is the detailed proof by using mathematical induction. According to
Lemma \ref{lem:guofaw}, for the cost $C_{0}$ and delay $D_{0}$
of the first phase, there exist $0\leq\alpha\leq2$, such that $C_{0}=\alpha C_{OPT}$
and $D_{0}=(2-\alpha)D$. Then since the algorithm didn't terminate
in the first iteration, we have $D_{0}>D$ and $C_{0}<C_{OPT}$. By induction,
$C_{i}<C_{OPT}$ holds. It remains only to prove $C_{i+1}<C_{OPT}$ for the case
$C_{i+1}>C_{i}$. That is, the computed bicameral cycle $O_{i}$ suffices
$c(O_{i})>0$. According to the definition of type-1 bicameral cycles,
we have
\begin{equation}
\frac{d(O_{i})}{c(O_{i})}=\frac{D_{i+1}-D_{i}}{C_{i+1}-C_{i}}\leq\frac{\Delta D_{i}}{\Delta C_{i}}=\frac{D-D_{i}}{C_{OPT}-C_{i}}.\label{eq:ratio}
\end{equation}
Before the $f$th iteration the delay of the solution is always larger
than $D$, so $D_{i+1}>D$ and $D_{i}>D$ both hold. Thus, we have:
\begin{equation}
0>D_{i+1}-D_{i}>D-D_{i}.\label{eq:di1di}
\end{equation}
 Combining Inequality (\ref{eq:ratio}) and Inequality (\ref{eq:di1di})
yields $C_{i+1}-C_{i}<C_{OPT}-C_{i}$, so $C_{i+1}<C_{OPT}$. This completes the
proof.
\end{proof}
The statement of this lemma is only for the case that the $k$RSP
problem is feasible. When the instance for $k$RSP is infeasible,
i.e., there does not exist $k$ disjoint paths satisfying the delay
constraint $D$, Algorithm \ref{alg:An-improved-algorithm} will detect
the infeasibility and return ``infeasible'' in Step 2(a).

\begin{figure*}
\begin{centering}
\includegraphics[width=0.9\textwidth]{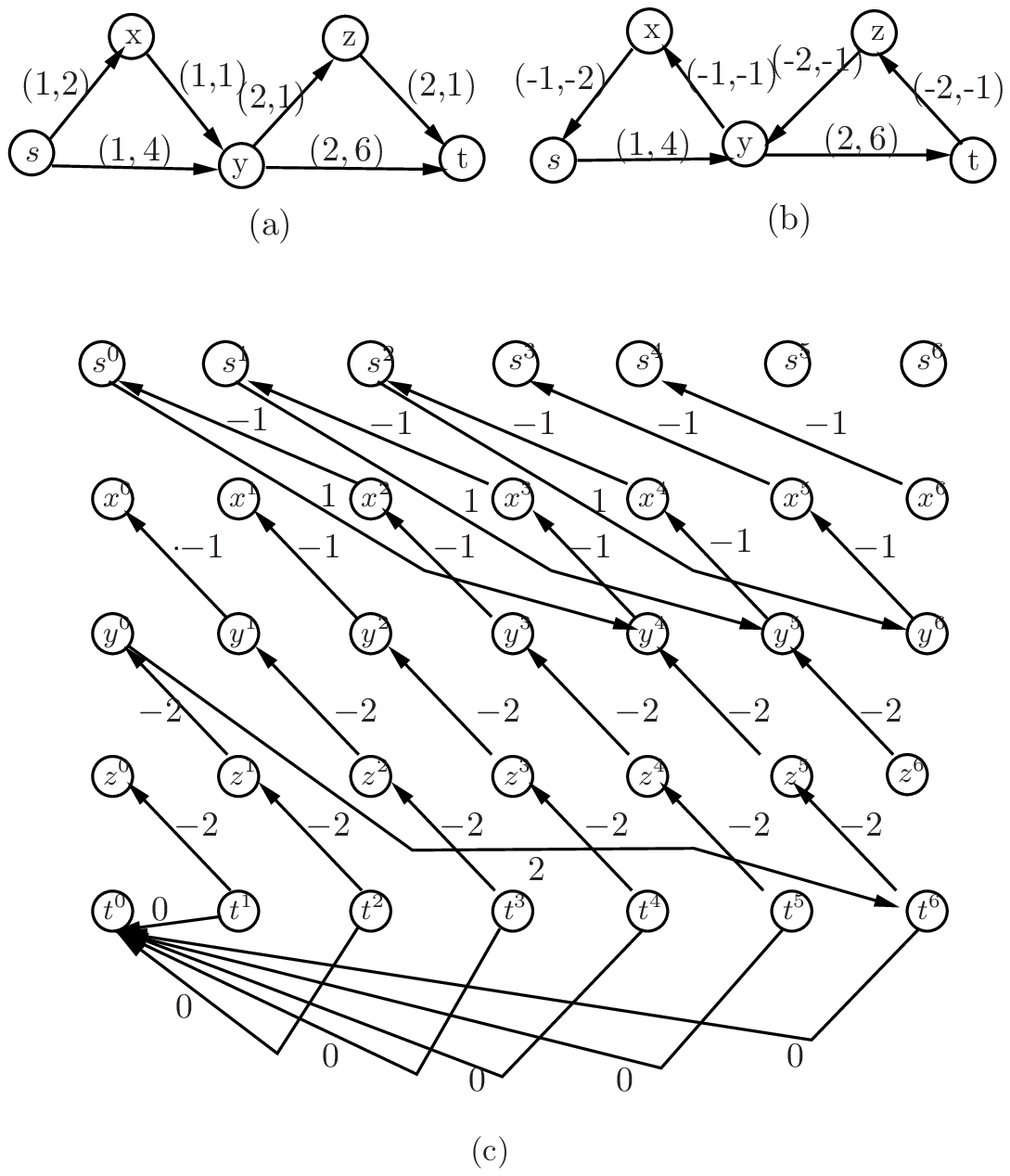}
\par\end{centering}

\protect\caption{\label{fig:Construction-of-acyclic}Construction of auxiliary graph
$H$ with cost constraint $B=6$. (a) graph $G$; (b) graph $\widetilde{G}$
wrt path $sxyzt$; (c) auxiliary graph $H$. }
\end{figure*}

Before presenting the time complexity analysis of Algorithm \ref{alg:An-improved-algorithm},
we would like first to investigate some properties of  bicameral
cycles,  as in the following lemma:
\begin{lem}
\label{lem:limofval}For the $i$th iteration, $i<f$, at least one
of the following two cases holds:
\begin{enumerate}
\item $r_{i+1}=r_{i}$ and $\Delta D_{i+1}<\Delta D_{i}$ ;
\item $r_{i+1}>r_{i}$.
\end{enumerate}
\end{lem}
\begin{proof}
Recall that $O_{i}$ is the bicameral cycle computed in the $i$th
iteration. If $O_{i}$ is a type-0 bicameral cycle, then $d(O_{i})<0$
and $-C_{OPT}\leq c(O_{i})\leq0$, or $d(O_{i})\leq0$ and $-C_{OPT}\leq c(O_{i})<0$
hold. For these two cases of type-0 bicameral cycle, Clause 2 obviously
holds. So we need only to consider type-1 and type-2 cycles, which
are as below:
\begin{enumerate}
\item $d(O_{i})<0$ and $0<c(O_{i})<C_{OPT}$

According Definition \ref{bicy}, $\frac{d(O_{i})}{c(O_{i})}\leq r_{i}$
holds. Then we have
\begin{equation}
d(O_{i})\leq r_{i}\cdot c(O_{i}).\label{eq:1}
\end{equation}
After the cycle cancellation wrt $O_{i}$, $r_{i+1}=\frac{\Delta D_{i}-d(O_{i})}{\Delta C_{i}-c(O_{i})}$.
Combining this inequality with Inequality (\ref{eq:1}) yields

\begin{equation}
r_{i+1}\geq\frac{\Delta D_{i}-r_{i}\cdot c(O_{i})}{\Delta C_{i}-c(O_{i})}=r_{i}\cdot\frac{\Delta C_{i}-c(O_{i})}{\Delta C_{i}-c(O_{i})}=r_{i}.\label{eq:case1}
\end{equation}

Since $d(O_{i})<0$, $\Delta D_{i+1}<\Delta D_{i}$ holds. So if the
two sides of Inequality (\ref{eq:case1}) are equal, Clause 1 holds;
otherwise Clause 2 holds.

\item $-C_{OPT}<c(O_{i})<0$ and $d(O_{i})>0$

The proof of the second case is similar to the first case. $O_{i}$
is with maximum $\frac{d(O_{i})}{c(O_{i})}$ according to Definition
\ref{bicy}, so $\frac{d(O_{i})}{c(O_{i})}>r_{i}$ holds, and hence

\end{enumerate}

\begin{equation}
d(O_{i})<r_{i}\cdot c(O_{i}).\label{eq:1-1}
\end{equation}

Then combining the above inequality with $r_{i+1}=\frac{\Delta D_{i}-d(O_{i})}{\Delta C_{i}-c(O_{i})}$,
we have

\[
r_{i+1}>\frac{\Delta D_{i}-r_{i}\cdot c(O_{i})}{\Delta C_{i}-c(O_{i})}=r_{i}\frac{\Delta C_{i}-c(O_{i})}{\Delta C_{i}-c(O_{i})}=r_{i}.
\]

Therefore, Clause 2 always holds for case 2. This completes the proof.
\end{proof}
Let $t_{bc}$ be the time of computing a bicameral cycle. We now consider
the time complexity of the algorithm. The key observation is that
$r_{i}$ is an arbitrary number. In fact, there are at most $\sum c(e)*\sum d(e)$
different values for $r$. That is because for the formula $r_{i}=\frac{\Delta D_{i}}{\Delta C_{i}}$,
the value of $\Delta D_{i}$ must be an integer between 0 and $-\sum d(e)$,
and so is $\Delta C_{i}$. Then according to Lemma \ref{lem:limofval},
the algorithm computes at most $\sum c(e)*\sum d(e)$ bicameral cycles
to decrease the value of $r_{i}$. On the other hand, the algorithm
computes at most $|D|*\sum c(e)*\sum d(e)$ bicameral cycles to decrease
the delay sum of the $k$ disjoint paths (Note that $\Delta D_{i}>\Delta D_{i+1}$
may hold). Therefore, we have the following lemma:
\begin{lem}
\label{lem:The-time-complexity}The time complexity of our algorithm
is $O(|D|*\sum c(e)*\sum d(e)*t_{bc})$.
\end{lem}
The time complexity of our algorithm can be analyzed better. However,
we omit the detailed and sophisticated analysis, because the length limit, and the focus of this paper should be on the approximation ratio.

\section{Computing Bicameral Cycle $O_{i}$ }

This section will show how to compute a bicameral cycle $O_{i}$.
This task is not easy, because almost every task involving bicriteria
negative cycle is \emph{NP}-hard. The key idea of our algorithm is
to construct two auxiliary graphs $H_{v}^{+}(B)$ and $H_{v}^{-}(B)$,
such that every cycle containing $v$  with total cost between
$0$ and $B$ in $\widetilde{G}$ is corresponding to a cycle in $H_{v}^{+}(B)$
while every cycle containing $v$  with cost between $-B$
and $0$ corresponds to a cycle in $H_{v}^{-}(B)$, where $B$ is
a given bound on cost and $v$ is a vertex of $\widetilde{G}$. Then
by computing the cycles in $H_{v}^{+}(B)$ and $H_{v}^{-}(B)$ for
every $v\in\widetilde{G}$ and every necessary $B$, we could
find in $\widetilde{G}$ bicameral cycles if there exists any. The
construction of auxiliary graph will be given in Subsection 4.1 and
the other parts of the algorithm shall be given in Subsection 4.2.

\subsection{Construction of Auxiliary Graph $H_{v}(B)$}

The algorithm of constructing the auxiliary graph $H_{v}^{+}(B)$ and
$H_{v}^{-}(B)$ is inspired by the method of modeling a shallow-light
spanning tree path subject to multiple constraints \cite{gouveia2011modeling}.
The full layout of the construction $H_{v}^{+}(B)$ is as shown in
Algorithm \ref{alg:Construction-of-auxiliary} (An example of such
a construction is as depicted in Figure \ref{fig:Construction-of-acyclic}).

\begin{algorithm}
\textbf{Input}: Residual graph $\widetilde{G}$ with edge cost $c:\, e\rightarrow\mathbb{Z}_{0}^{+}$
and edge delay $d:\, e\rightarrow\mathbb{Z}_{0}^{+}$, a specified
vertex $v\in V$, a given cost constraint $B\in\mathbb{Z}^{+}$;

\textbf{Output}: Auxiliary graph $H_{v}^{+}(B)$.
\begin{enumerate}
\item \textbf{For} every vertex $v_{l}$ of $V$ \textbf{do}

\quad{}Add $B+1$ vertices $v_{l}^{0},\dots,v_{l}^{B}$ to $H_{v}^{+}(B)$
;

\item \textbf{For} every edge $e=\left\langle v_{j},v_{l}\right\rangle \in E$
\textbf{do}

\begin{enumerate}
\item if $c(e)\geq0$, add to $H_{v}^{+}(B)$ the edges $\left\langle v_{j}^{0},\, v_{l}^{c(e)}\right\rangle ,\,\dots,\left\langle v_{j}^{B-c(e)},\, v_{l}^{B}\right\rangle $,
each of which is with delay $d(e)$;
\item if $c(e)<0$, add to $H_{v}^{+}(B)$ the edges $\left\langle v_{j}^{B},\, v_{l}^{B-|c(e)|}\right\rangle ,\,\dots,\left\langle v_{j}^{|c(e)|},\, v_{l}^{0}\right\rangle $,
each of which is with delay $d(e)$.
\end{enumerate}
\item \textbf{For} each vertex $v^{i}$ of $H_{v}^{+}(B)$ that corresponding
to the specified vertex $v\in G$ \textbf{do}

\quad{}Add edge $e(v^{i},\, v^{0})$ to $H_{v}^{+}(B)$ with delay
zero.
/{*} To construct $H_{v}^{-}(B)$ is to add edge $e(v^{i},\, v^{B})$. {*}/

\end{enumerate}
\protect\caption{\label{alg:Construction-of-auxiliary} Construction of auxiliary graph
$H_{v}^{+}(B)$.}
\end{algorithm}

\subsection{Computation of Bicameral Cycle $O_{i}$}

This subsection shall show how to use LP-rounding method to compute
cycles in $H_{v}^{+}(B)$ and $H_{v}^{-}(B)$, and obtain a bicameral
cycle in $\widetilde{G}$. We start from the following linear programming
formula:

\begin{equation}
\min\sum_{e\in H}c(e)x(e)\label{eq:lp}
\end{equation}

subject to

\[
\sum_{e\in\delta^{+}(v)}x_{e}-\sum_{e\in\delta^{-}(v)}x_{e}=\begin{array}{cc}
0 & \forall v\in V(H_{v}^{+}(B)\, or\, H_{v}^{-}(B))\end{array}
\]

\[
\sum_{e\in H}d(e)x(e)\leq\Delta D
\]

Unlike other linear programming formula for shortest paths, spanning
trees, or Steiner trees, $0\leq x(e)\leq1$ is not necessary for our
formula, because our goal is only to compute a bicameral cycle. Below
we shall discuss $H_{v}^{+}(B)$ only, since the case for $H_{v}^{-}(B)$
is similar.
\begin{lem}
\label{lem:LPsolution-tocycle}A solution to LP (\ref{eq:lp}) exactly
corresponds to a set of cycles with cost between $B$ and $-B$ in
residual graph $\widetilde{G}$.
\end{lem}
A solution to LP (\ref{eq:lp}) is apparently corresponding to a set
of fractional cycles of $H_{v}^{+}(B)$ (or $H_{v}^{-}(B)$). So we
need only to show that a cycle in $H_{v}^{+}(B)$ corresponds to a
set of cycles in $\widetilde{G}$. In fact, according to the construction
of $H_{v}^{+}(B)$ , we have the following lemma, from which the correctness
of Lemma \ref{lem:LPsolution-tocycle} can be immediately obtained.
\begin{lem}
\label{lem:cyclecost} A cycle in $H_{v}^{+}(B)$ corresponds to a
set of cycles in $\widetilde{G}$, each of which is with cost between
$B$ and $-B$. Conversely, a cycle in $\widetilde{G}$, containing
$v$ and with cost between $0$ and $B$, corresponds to a cycle in
$H_{v}^{+}(B)$.\end{lem}
\begin{proof}
Assume that $O$ is a cycle in $H_{v}^{+}(B)$. According to the rules
of adding the edges of $H_{v}^{+}(B)$, every edge in $H_{v}^{+}(B)$
either corresponds to an edge in $\widetilde{G}$ or is between duplications
of an identical vertex of $\widetilde{G}$. So $O$ corresponds to
a closed walk in $\widetilde{G}$, and hence it corresponds to a set
of cycles of $\widetilde{G}$. The cost bound of the set of cycles
follows from the construction of $H_{v}^{+}(B)$ similarly.

Conversely, let $O'$ be a cycle in $\widetilde{G}$, with cost between
$0$ and $B$, i.e., $0\leq c(O')\leq B$. Start from $v^{0}\in H_{v}^{+}(B)$,
a vertex corresponding to $v\in O'$, according to the edges on $O'$,
we can find a path from $v^{0}$ to $v^{c(O')}$, since $c(O')\leq B$.
The path plus the edge $e(v^{c(O')},v^{0})$ exactly compose a cycle
in $H_{v}^{+}(B)$.\end{proof}
\begin{thm}
\label{thm:corethrcycles-}Let cycles $O_{1},\dots,O_{z}$ be the
cycles in $\widetilde{G}$ corresponding to $\chi(v,B^{*})$ for every
$v\in\widetilde{G}$, where $B^{*}=\sum_{i=1}^{k}c(P_{i}^{*})$ and
$P_{1}^{*},\dots,P_{k}^{*}$ is an optimal solution to the $k$RSP
problem. Then there must be a bicameral cycle in $O_{1},\dots,O_{z}$
if the $k$RSP problem is feasible. \end{thm}
\begin{proof}
If there exist $c(O_{i})\leq0$ and $d(O_{i})<0$, or $d(O_{i})\leq0$
and $c(O_{i})<0$, then $O_{i}$ is type-0 bicameral cycle. So without
loss of generality, we can assume that either $c(O_{i})\geq0$ or
$d(O_{i})\geq0$ holds. Let $\chi^{+}(v,\, B^{*})$ be an optimal
solution to LP (\ref{eq:lp}) against $H_{v}^{+}(B)$.

Let $P_{1},\dots,P_{k}$ be the current solution against $k$RSP.
Let $\{O_{1}^{*},\dots,O_{h}^{*}\}=\{P_{1}^{*},\dots,P_{k}^{*}\}\oplus\{\overline{P_{1}},\dots,\overline{P_{k}}\}$.
Then apparently
\[
-\sum_{i=1}^{k}c(P_{i}^{*})\leq-\sum_{i=1}^{k}c(P_{i})\leq c(O_{i}^{*})\leq\sum_{i=1}^{k}c(P_{i}^{*}).
\]

That is, cycle $O_{i}^{*}$ with $d(O_{i}^{*})<0$ and $c(O_{i}^{*})>0$
can find its corresponding cycle in $\cup_{v\in\widetilde{G}}H_{v}^{+}(B^{*})$.
Without loss of generality assume that $O_{i}^{*}$ with minimum $\frac{d(O_{i}^{*})}{c(O_{i}^{*})}$
among all other cycles of $\{O_{1}^{*},\dots,O_{h}^{*}\}$ with negative
delay, and its corresponding cycle is in $H_{v_{j}}(B^{*})$. Then
we have $\frac{d(\chi(v_{j},B^{*}))}{c(\chi(v_{j},B^{*}))}\leq\frac{d(O_{i}^{*})}{c(O_{i}^{*})}$.
So in the set of cycles of $\widetilde{G}$ that corresponding to
$\chi^{+}(v,B^{*})$, there must exist a cycle $O_{1}$ with $\frac{d(O_{1})}{c(O_{1})}\leq\frac{d(O_{i}^{*})}{c(O_{i}^{*})}$
and $d(O_{1})<0$, or $\frac{d(O_{2})}{c(O_{2})}\geq\frac{d(O_{i}^{*})}{c(O_{i}^{*})}$
and $c(O_{2})<0$. The case for $O_{i}^{*}$ with $d(O_{i}^{*})\geq0$,
$c(O_{i}^{*})<0$ and maximum $\frac{d(O_{i}^{*})}{c(O_{i}^{*})}$
is similar in $\cup_{v\in\widetilde{G}}H_{v}^{-}(B^{*})$, and hence
we can find $O'_{1}$ and $O'_{2}$ accordingly. Then from the definition
of bicameral cycle, at least one cycle of $\{O_{1},O'_{1},O_{2},O'_{2}\}$
is a bicameral cycle. This completes the proof.
\end{proof}
Following Theorem \ref{thm:corethrcycles-}, the main steps of our
algorithm roughly proceeds as:
\begin{enumerate}
\item Construct $H_{v}^{+}(B)$ and $H_{v}^{-}(B)$ for all $v\in\widetilde{G}$
and every integer $0\leq B\leq\sum_{e\in G}c(e)$;
\item Solve LP (\ref{eq:lp}) against every $H_{v}^{+}(B)$ and $H_{v}^{-}(B)$,
and collect optimal solutions;
\item Compute a bicameral cycle among all the cycles released from all the
solutions.
\end{enumerate}
In general, the above algorithm follows to the LP-rounding algorithm
framework. Let's revise the algorithm as well as its proofs following
the traditional line of analyzing a LP-rounding based approximation
algorithm. Clearly, the algorithm solves LP formulas, and rounds $x_{e}$
to $1$ for every edge $e$ that belongs to a computed bicameral cycle.
Then the core task in the analysis is to show that the cycles corresponding to
the solution of LP (\ref{eq:lp}) always contains a bicameral cycle,
if the $k$RSP problem is feasible. This task has actual been done
in Theorem \ref{thm:corethrcycles-}.

The detailed algorithm is as in Algorithm \ref{alg:compbicameracycle}.

\begin{thm}
\label{thm:bicyctime}Algorithm \ref{alg:compbicameracycle} correctly
computes a bicameral cycle in time $t_{bc}=O(n^{4.5}C_{OPT}^{4.5}\sum c(e)L)$.\end{thm}
\begin{proof}
Following Algorithm \ref{alg:compbicameracycle}, the resulting cycle
satisfies the definition of bicameral cycle. So it remains only to show the
time complexity of the algorithm. Step 1 of the algorithm constructs
$O(V(H_{v}^{+}(B))*\sum c(e))=O(nC_{OPT}\sum c(e))$ auxiliary graphs, each
takes $O(V(H_{v}^{+}(B))^{3.5}L)=O(n^{3.5}C_{OPT}^{3.5}L)$ to solve the
linear programming formula, where $L$ is the maximum length of the
input \cite{korte2012combinatorial}. Other steps of the algorithm
take trivial time comparing to Step 1. Therefore, the time complexity
of the algorithm is $O(n^{4.5}C_{OPT}^{4.5}\sum c(e)L)$.
\end{proof}
The above time complexity is terrible, and can be significantly improved
by some sophisticated algorithm design techniques. For example, construction of auxiliary
graphs for all \textbf{$B=1$ }to\textbf{ $\sum c(e)$} is not necessary.
Binary search can be applied here to find $B^{*}$, and reduce the
number of the auxiliary graphs constructed. Due to the length limit,
we omit the details here. After all, the core of this paper is to
show that the $k$RSP problem admits an approximation ratio of $(1+\epsilon,\,2+\epsilon)$
for any $\epsilon>0$.

\begin{algorithm}
\textbf{Input}: Graph $\widetilde{G}$;

\textbf{Output}: A bicameral cycle $O$ in $\widetilde{G}$.
\begin{enumerate}
\item \textbf{For $B=1$ }to\textbf{ $\sum c(e)$ do}

\begin{enumerate}
\item \textbf{For} each $v\in\widetilde{G}$ \textbf{do}

\begin{enumerate}
\item Construct $H_{v}^{+}(B)$;
\item Solve LP (\ref{eq:lp}) against $H_{v}^{+}(B)$, and obtain two optimal
solution $\chi^{+}(v,B)$;
\item Release the set of cycles $\mathbb{O}^{+}(v,B)$ in $\widetilde{G}$,
which corresponds to $\chi^{+}(v,B)$;
\item \textbf{If} there exists $O\in\mathbb{O}^{+}(v,B)$ with $c(O)\leq0$
and $d(O)<0$ or $c(O)<0$ and $d(O)\leq0$, \textbf{then} terminate;
\end{enumerate}

/{*} $O$ is a type-0 bicameral cycle. {*}/

\item \textbf{For} each $v\in\widetilde{G}$ \textbf{do}

\begin{enumerate}
\item Construct $H_{v}^{-}(B)$;
\item Obtain the set of cycles $\mathbb{O}^{-}(v,B)$ , and check if it
contains type-0 bicameral cycles;
\end{enumerate}

/{*} $O$ is a type-0 bicameral cycle. {*}/

\end{enumerate}
\item Compute $O_{1}$ with minimum $\frac{d(O_{1})}{c(O_{1})}$ and $d(O_{1})<0$,
and $O_{2}$ with minimum $\frac{d(O_{2})}{c(O_{2})}$ and $c(O_{2})<0$
among all the cycles of $\underset{v}{\cup}\underset{B}{\cup}(\mathbb{O}^{+}(v,B)\cup\mathbb{O}^{-}(v,B))$;
\item \textbf{If} $\vert\frac{d(O_{1})}{c(O_{1})}\vert\leq\vert\frac{d(O_{2})}{c(O_{2})}\vert$
\textbf{then}

\quad{}return $O_{1}$; /{*} $O_{1}$ is a type-1 bicameral cycle.
{*}/

\textbf{Else }return $O_{2}$. /{*} $O_{2}$ is a type-2 bicameral
cycle. {*}/

\end{enumerate}
\protect\caption{\label{alg:compbicameracycle}Computation of bicameral cycle $O$.}
\end{algorithm}

\section{Conclusion}

For any constant $\epsilon>0$, this paper first gave a polynomial time
approximation algorithm with bifactor approximation ratio
$(1+\epsilon,\,2+\epsilon)$, based on improving an approximate solution with
bifactor ratio $(2,2)$ via using bicameral cycles in the cycle cancellation
method. Then this paper presented a constructive method for computing a
bicameral cycle, by constructing an auxiliary graph and innovatively employing
LP-rounding technique therein. To the best of our knowledge, our algorithm is
the first constant factor approximation algorithm that computes a solution
almost strictly obeying the delay constraint. We are now investigating the
inapproximability of the $k$RSP problem. \

\bibliographystyle{plain}
\bibliography{disjointQoS}

\end{document}